\documentclass[conference]{IEEEtran}

\usepackage{cite}

\usepackage{amsmath}
\usepackage{amssymb}
\usepackage{amsthm}
\usepackage{bbm}
\usepackage{mathtools}

\usepackage{array}
\usepackage{tikz}
\usepackage{pgfplots}
\pgfplotsset{compat=newest} 
\pgfplotsset{plot coordinates/math parser=false}
\usepackage{graphicx}

\usetikzlibrary{decorations.markings}
\usetikzlibrary{matrix}

\newif\iffull
\fulltrue

\newcommand\blfootnote[1]{%
  \begingroup
  \renewcommand\thefootnote{}\footnote{#1}%
  \addtocounter{footnote}{-1}%
  \endgroup
}

\newcommand{\B}{{\rm\scriptscriptstyle  B}}
\newcommand{\BP}{{\rm\scriptscriptstyle  BP}}
\newcommand{\eps}{{\varepsilon}}

\newtheorem{theorem}{Theorem}

\begin{document}

\title{Spatially Coupled LDPC Codes Affected by a Single Random Burst of Erasures}

\author{\IEEEauthorblockN{Vahid Aref\IEEEauthorrefmark{1}, Narayanan Rengaswamy\IEEEauthorrefmark{2}, and Laurent Schmalen\IEEEauthorrefmark{1}}
\IEEEauthorblockA{
\IEEEauthorrefmark{1} Nokia Bell Labs, Stuttgart, Germany. (e-mail: \{\texttt{firstname.lastname}\}\texttt{@nokia-bell-labs.com})
\\
\IEEEauthorrefmark{2} Department
of Electrical and Computer Engineering, Duke University, Durham,
NC, 27708 USA}
}

\maketitle

\begin{abstract}
Spatially-Coupled LDPC (SC-LDPC) ensembles achieve the capacity of binary memoryless channels (BMS), asymptotically, under belief-propagation (BP) decoding. 
In this paper, we study the BP decoding of these code ensembles over a BMS channel and in the presence of a single random burst of erasures. 
We show that in the limit of code length, codewords can be recovered successfully if 
the length of the burst is smaller than some maximum recoverable burst length. 
We observe that the maximum recoverable burst length 
is practically the same if the transmission takes place over binary erasure channel or over binary additive white Gaussian channel with the same capacity. 
Analyzing the stopping sets,
we also estimate the decoding failure probability (the error floor) when the code length is finite.
\end{abstract}


\section{Introduction}
\label{sec:intro}

Low-density parity-check (LDPC) codes are widely used due to
their outstanding performance under low-complexity belief propagation (BP) decoding.  
However, an error probability exceeding that of
maximum-a-posteriori (MAP) decoding has to be tolerated with (sub-optimal) BP decoding. 
Recently, it has been empirically observed 
for spatially coupled LDPC (SC-LDPC) codes
--- first introduced  
as convolutional LDPC codes 
---
that the BP performance of 
these codes can improve dramatically towards the MAP performance of the 
underlying LDPC code under many different settings and conditions, e.g.~\cite{Lentmaier-ita09}.
This phenomenon, termed \emph{threshold saturation}, has been proven rigorously in \cite{Kudekar-it11,Kudekar-it13}. 
In particular, the BP threshold of a coupled LDPC ensemble tends to its MAP threshold on any binary memoryless symmetric
channel (BMS). 
\blfootnote{Parts of this work were conducted while N. Rengaswamy was visiting Bell Labs as a research intern funded by a scholarship of the DAAD-RisePro programme. The work of L. Schmalen was funded by the German Government in the frame of the CELTIC+/BMBF project SASER-SaveNet.}




Besides their excellent performance on the BEC and AWGN channels, much less is known about the burst error correctability of SC-LDPC codes. 
In~\cite{Jule-isit13}, SC-LDPC ensembles over a block erasure channel (BLEC) are considered with a channel that erases a complete spatial positions instead of individual bits. 
This block erasure model mimics block-fading channels frequently occurring in wireless communications. 
The authors give asymptotic lower and upper bounds for the bit and block erasure probabilities obtained from density evolution. 
Protograph-based codes that maximize the correctable burst lengths are constructed in~\cite{Iyengar-icc10}, while interleaving (therein denoted band splitting) is applied to a protograph-based SC-LDPC code in~\cite{Mori-corr15 } to increase the correctable burst length. 
If windowed decoding is used, this approach results however in an increased required window length and thus also in an increased complexity. 
Recently, it has been shown that protograph-based LDPC codes can increase the diversity order of block fading channels and are thus good candidates for block erasure channels~\cite{ulHassan-isit14,ulHassan-itw15}; however, they require large syndrome former memories if the burst length becomes large. Closely related structures based on protographs have been proposed in~\cite{Jardel-comnet15} which spatially couple the special class of root-check LDPC codes~\cite{Boutros-tit10} to  improve the finite length performance and thresholds.
 
In this paper, we are interested in the burst correction capabilities of general spatially coupled LDPC code ensembles as introduced in~\cite{Kudekar-it11}. We chose this ensemble as we know that it is capacity-achieving for BMS channels and therefore likely to be picked as potential candidate for various communication systems. We are in particular interested in knowing if besides their excellent performance on BMS channels, these codes also have advantages when subject to burst errors. In this paper, we extend our results of~\cite{RengaswamyZSC16}, where we derived tight lower bounds on the correctability of a long burst of erasures erasing either a complete spatial position or slightly more. In this paper, we investigate the maximum length of the correctable burst by utilizing density evolution to find thresholds on the correctable bursts in the asymptotic block-length regime. Additionally, we find expressions for the expected error floor in the non-asymptotic regime by counting small-size stopping sets. Finally, we verify all findings in a simulation example.

\section{Preliminaries}
\subsection{The Regular $(d_v,d_c,w,L,M)$ SC-LDPC Ensemble}
\label{sec:randomSCLDPC}

We now briefly review how to sample a code from a random regular $\mathcal{C}_{\mathcal{R}}(d_v,d_c,w,L,M)$ SC-LDPC ensemble~\cite{Kudekar-it11}. 
We first lay out a set of positions indexed from $z=1$ to $L$ on a \emph{spatial dimension}. 
At each spatial position (SP), $z$, there are $M$ variable nodes (VNs) and $M\frac{d_v}{d_c}$ check nodes (CNs), 
where $M\frac{d_v}{d_c} \in \mathbb{N}$ and, $d_v$ and $d_c$ denote the variable and check node degrees, respectively.
Let $w>1$ denote the smoothing (coupling) parameter. 
Then, we additionally consider $w-1$ sets of $M\frac{d_v}{d_c}$ CNs in SPs $L+1,\dots,L+w-1$. 
Every CN is equiped with $d_c$ ``sockets'' and imposes an even parity constraint on its $d_c$ neighboring VNs, connected via the sockets. 
Each VN in SP $z$ is connected to $d_v$ CNs in SPs $z,\dots,z+w-1$ as follows: 
each of the $d_v$ edges of this VN is allowed to randomly and uniformly connect to any of the $wMd_v$ sockets arising from the CNs in SPs $z,\dots,z+w-1$, such that parallel edges are avoided in the resulting bipartite graph. We avoid parallel edges as it turns out that for practical finite $M$, the presence of parallel edges can have detrimental effects on the models that we consider.
This graph represents the code so that we have $N=LM$ code bits, distributed over $L$ SPs. 
Note that the CNs at the boundary SPs, i.e., at SPs $1,\ldots,w-1$ and $L+1,\ldots,L+w-1$, can have degree less than $d_c$, due to the termination of the code and the absence of VNs outside SPs $1,\ldots,L$. 
Zero degree CNs are removed from the code.
Because of additional check nodes in SPs $z>L$, the code rate amounts $r = 1-\frac{d_v}{d_c}-\delta$, where $\delta=O(\frac{w}{L})$.
{Throughout this work, we assume the two mild conditions of $d_v \geq 3$ and $wM \geq 2(d_v+1)d_c$.}

\subsection{Burst Error Channel Model}

Due to impairments such as slow fading, carrier phase or frequency noise, the loss of a data frame, or the outage of a node in distributed storage,
a number of sequential received code bits may be severely distorted or erased. Depending on the channel of interest,
several error bursts with different lengths may occur
in a codeword. As a building structure of different models, 
we consider in this paper a single error burst of length $B=bM$ with a randomly chosen starting position. 
For simplicity, we assume that these bits are erased by the channel. Additionally, we assume that the transmission takes place over the
binary erasure channel (BEC), or over the binary-input additive white Gaussian noise (BiAWGN) channel. In a BEC($\eps$), the received bits outside of the burst are randomly
erased with probability $\eps$, otherwise received correctly.
      
\section{The Maximum Burst Length in the Asymptotic Block-length Regime}
 \label{sec:asymp}
Consider a $(d_v,d_c,w,L,M)$ SC-LDPC code used for transmission over a BEC($\eps$). Additionally, 
a random block of consecutive code bits is erased by
a burst of length $B=bM$. The starting bit of the burst, $S$, is uniformly chosen from code bits $[1,LM-bM+1]$.
For a given $S$, let $m_z$ denote the number of code bits erased by the burst and belonging to spatial position $z$.
Define $s=\frac{S}{M}$ and $z_0=\lceil s\rceil$. Then $(z_0-1)M< S\leq z_0M$ and 
\begin{align*}
m_z\! &=\! \left\{\!\!\begin{array}{ll}
0 & z < z_0 \\
\min\{bM,z_0M-sM+1\} & z = z_0 \\
\max\{0,\min\{M,bM\!+\!sM\!-1\!-(z\!-\!1)M\}\}& z>z_0\end{array}\right.
\end{align*}
Therefore, $\eps_z = \eps + \frac{m_z}{M}(1-\eps)$ is the average erasure probability of code bits in
spatial position $z$.
We use density evolution (DE) to evaluate the asymptotic performance of the code
ensemble $\mathcal{C}_{\mathcal{R}}(d_v,d_c,w,L,M)$ under BP decoding when $M\to\infty$~\cite{Kudekar-it11,Kudekar-it13}. This method estimates how the
empirical distribution of the log-likelihood ratio (LLR) of code bits at each position
$z$ evolves iteratively during BP decoding, given the empirical
distribution of the received bits' LLRs. For transmission over erasure channels, the
LLR distribution can be represented by a scalar value, the
erasure probability $\eps_z$, and the DE equation turns into a scalar
update recursion. In that case, the update equation becomes
\begin{equation*}
x^{(t+1)}_z = \eps_z \left(1-\frac{1}{w}\sum_{i=0}^{w-1}\left(1-\frac{1}{w}\sum_{j=0}^{w-1} x^{(t)}_{z+i-j}\right)^{d_c-1} \right)^{d_v-1}
\end{equation*}
where $x^{(t)}_z$ denote the average erasure probability of the outgoing messages from code bits in position $z$ and at iteration $t$. We initialize $x_z^{(0)}=1$ for all $z\in[1,L]$ and $x_z^{(t)}=0, t\geq 0$ otherwise. For a given $s$ and $b$, we hence have
\begin{equation*}
\eps_z=
\begin{cases}
\eps & z<\lceil s\rceil\\
\eps + (1-\eps)\min\{b,\lceil s\rceil -s\} & z=\lceil s\rceil\\
\eps + (1-\eps)\max\{0,\min\{1, b+s-z+1\}\} & z>\lceil s\rceil
\end{cases}
\end{equation*}

The average probability that a code bit is not recovered after $T$ iterations is given by
\begin{equation*}
P_{\rm e}(T,b,s) = \frac{1}{L}\sum_{z=0}^{L-1}\eps_z\! \left(1\!-\!\frac{1}{w}\sum_{i=0}^{w-1}(1\!-\!\frac{1}{w}\sum_{j=0}^{w-1} x^{(T)}_{z+i-j})^{d_c-1} \right)^{d_v}
\end{equation*}
Since for $M\to\infty$, $s$ is uniformly random over $[0,L-b]$, we have
\begin{equation*}
P_{\rm e}(T,b,\eps) =\frac{1}{L-b}\int_{0}^{L-b} P_{\rm e}(T,b,s)\text{d}s.
\end{equation*}
 We define the largest recoverable burst length $b_{\BP}$ as follows:
 \begin{equation}
 \label{eq:b_bp_bec}
 b_{\BP}(\eps)=\sup\{b\mid b>0, \lim_{T\to\infty} P_{\rm e}(T,b,\eps)=0\}.
 \end{equation}
 
We numerically compute $b_{\BP}(\eps)$ for the two ensembles 
$\mathcal{C}_{\mathcal{R}}(3,6,w,L)$ and $\mathcal{C}_{\mathcal{R}}(4,8,w,L)$ with $w=3,4,5$ and $L\gg w$. 
For a given $b$, we run DE and evaluate $P_{\rm e}(T,b,s)$
over all $s=k\Delta$, where $k\in\mathbb{N}$ and $\Delta=0.001$. The number of iterations $T$ is limited by the following stopping criterion:
\begin{equation*}
T=\min\left\{t\ \middle|\  \frac{1}{L}\sum_{z=1}^L |x_z^{(t-1)}-x_z^{(t)}|<10^{-5}\right\}.
\end{equation*}

Fig.~\ref{fig:b_bp36} shows $b_{\BP}(\eps)$ for $\mathcal{C}_{\mathcal{R}}(3,6,w,L)$ and $\mathcal{C}_{\mathcal{R}}(4,8,w,L)$. We observe that $b_{\BP}(\eps)$ is decreasing in terms of $\eps$ and it becomes zero
at $\eps_\BP^{(3,6,w,L)}\approx0.488$ and $\eps_\BP^{(4,8,w,L)}\approx0.497$. As one may expect, a longer burst can be recovered as the gap to the BP threshold of the ensemble on a BEC without bursts increases. Moreover, we observe that $b_{\BP}(\eps)$ is increasing in $w$ but is decreasing as $d_c$ increases. 

\begin{figure}[tb!]
\begin{tikzpicture}
\begin{scope}

\begin{axis}[%
width=1.0\columnwidth,
height=0.65\columnwidth,
every axis/.append style={font=\small},
every x tick label/.append style={font=\footnotesize},
every y tick label/.append style={font=\footnotesize},
xlabel={$\varepsilon, 1-C(N_0)$},
xmajorgrids,
xmin=0,
xmax=.5,
ymin=0,
ymax=3,
yminorticks=true,
ylabel={$b_{\rm\scriptscriptstyle  BP}$},
ymajorgrids,
yminorgrids,
legend style={legend cell align=left,align=left,draw=black,font=\scriptsize},
]

\definecolor{mycolor1}{rgb}{0.00000,0.44700,0.74100}
\definecolor{mycolor2}{rgb}{0.85000,0.32500,0.09800}
\definecolor{mycolor3}{rgb}{0.92900,0.69400,0.12500}
\definecolor{mycolor4}{rgb}{0.49400,0.18400,0.55600}
\definecolor{mycolor5}{rgb}{0.46600,0.67400,0.18800}

\addplot [color=mycolor1,thick,solid,mark=none]
table[x=eps,y=bec363] {de_rbc_bec_36.txt};
\addlegendentry{ BEC, $w=3$};

\addplot [color=mycolor2,thick,solid,mark=none]
table[x=eps,y=bec364] {de_rbc_bec_36.txt};
\addlegendentry{ BEC, $w=4$};

\addplot [color=mycolor3,thick,solid,mark=none]
table[x=eps,y=bec365] {de_rbc_bec_36.txt};
\addlegendentry{  BEC, $w=5$};

\addplot [color=mycolor1,thick,dashed,mark=none]
table[x=h,y=bawgn363] {de_rbc_bawgn_36.txt};
\addlegendentry{  BiAWGN, $w=3$};

\addplot [color=mycolor2,thick,dashed,mark=none]
table[x=h,y=bawgn364] {de_rbc_bawgn_36.txt};
\addlegendentry{ BiAWGN, $w=4$};

\addplot [color=mycolor3,thick,dashed,mark=none]
table[x=h,y=bawgn365] {de_rbc_bawgn_36.txt};
\addlegendentry{ BiAWGN, $w=5$};

\node at (.05,.2) {$(a)$};

\end{axis}

\end{scope}

\begin{scope}[yshift=-0.59\columnwidth]
\begin{axis}[%
width=1.0\columnwidth,
height=0.65\columnwidth,
every axis/.append style={font=\small},
every x tick label/.append style={font=\footnotesize},
every y tick label/.append style={font=\footnotesize},
xlabel={$\varepsilon, 1-C(N_0)$},
xmajorgrids,
xmin=0,
xmax=.5,
ymin=0,
ymax=3,
yminorticks=true,
ylabel={$b_{\rm\scriptscriptstyle  BP}$},
ymajorgrids,
yminorgrids,
legend style={legend cell align=left,align=left,draw=black,font=\scriptsize},
]

\definecolor{mycolor1}{rgb}{0.00000,0.44700,0.74100}
\definecolor{mycolor2}{rgb}{0.85000,0.32500,0.09800}
\definecolor{mycolor3}{rgb}{0.92900,0.69400,0.12500}
\definecolor{mycolor4}{rgb}{0.49400,0.18400,0.55600}
\definecolor{mycolor5}{rgb}{0.46600,0.67400,0.18800}

\addplot [color=mycolor1,thick,solid,mark=none]
table[x=eps,y=bec483] {de_rbc_bec_48.txt};
\addlegendentry{ BEC, $w=3$};

\addplot [color=mycolor2,thick,solid,mark=none]
table[x=eps,y=bec484] {de_rbc_bec_48.txt};
\addlegendentry{ BEC, $w=4$};

\addplot [color=mycolor3,thick,solid,mark=none]
table[x=eps,y=bec485] {de_rbc_bec_48.txt};
\addlegendentry{  BEC, $w=5$};

\addplot [color=mycolor1,thick,dashed,mark=none]
table[x=h,y=bawgn483] {de_rbc_bawgn_48.txt};
\addlegendentry{  BiAWGN, $w=3$};

\addplot [color=mycolor2,thick,dashed,mark=none]
table[x=h,y=bawgn484] {de_rbc_bawgn_48.txt};
\addlegendentry{ BiAWGN, $w=4$};

\addplot [color=mycolor3,thick,dashed,mark=none]
table[x=h,y=bawgn485] {de_rbc_bawgn_48.txt};
\addlegendentry{ BiAWGN, $w=5$};

\node at (.05,.2) {$(b)$};
\end{axis}

\end{scope}

\end{tikzpicture}%
\caption{\label{fig:b_bp36} The BP recoverable burst length $b_{\BP}$ for $(a)$ $\mathcal{C}_\mathcal{R}(3,6,w,L)$ and $(b)$ $\mathcal{C}_\mathcal{R}(4,8,w,L)$ ensembles. Solid lines are when the transmission is over the BEC, and dashed lines are when the transmission is over the BiAWGN.}
\end{figure}
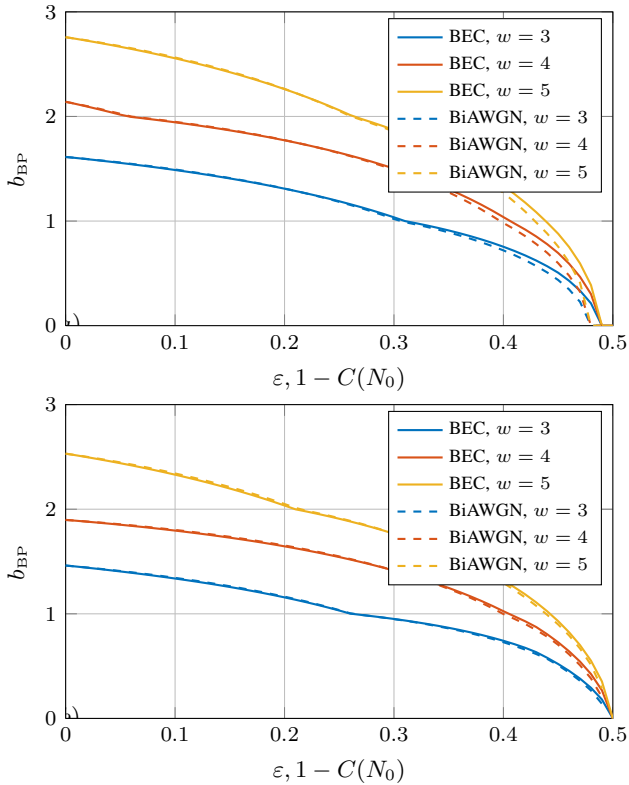

Now we consider the transmission over the BiAWGN channel with an
additional burst of erasures. We assume that the received bits have signal-to-noise ratio (SNR) $10\log_{10}(2/N_0)$ dB.
We again use DE to compute the recoverable burst length $b_{\BP}$ when $M\to\infty$. The DE equations for
SC-LDPC codes over a BMS channel are detailed in \cite{Kudekar-it13}. For a given $s$ and $b$, 
the received bits in spatial position $z$ are erased with probability $\frac{m_z}{M}$, or distorted by the Gaussian noise. Thus, the LLR distribution of received bits in each spatial position is the convex combination of two LLR distributions: the  distribution of BiAWGN channel and the distribution of erased bits. 

For a given SNR, we can define $b_{\BP}(N_0)$ similar to~\eqref{eq:b_bp_bec}. 
To numerically compute $b_{\BP}(N_0)$, we use the DE method of~\cite[App. B]{Richardson-MCT08} in which 
the quantized LLR distributions are updated recursively. 
 For a given $b$,
we run DE over all $s=k\Delta$, and $\Delta=0.01$. We also use a similar stopping criterion as for the BEC.

Fig.~\ref{fig:b_bp36} also shows $b_{\BP}(N_0)$ for $\mathcal{C}_{\mathcal{R}}(3,6,w,L)$ and $\mathcal{C}_{\mathcal{R}}(4,8,w,L)$. To have a fair comparison with BEC, we plot $b_{\BP}(N_0)$ in terms of $1-C(N_0)$, where $C(N_0)$ is the capacity of BiAWGN channel with SNR $10\log_{10}(2/N_0)$ dB. We observe that $b_{\BP}(\eps)$ and 
$b_{\BP}(N_0)$ are \textit{almost} equal when $C(N_0)=1-\eps$. Note that the deviation of both curves for SNR values close to the BP threshold is mainly because the quantization level of  the LLR distributions was not small enough for those SNR values. However, it is very unlikely that $b_{\BP}(\eps)$ and $b_{\BP}(N_0)$ are exactly equal as the BP threshold of these ensembles over the BEC and BiAWGN channel without burst errors are not equal either (but they are very close).  

\textbf{Remark 1:} We defined $b_{\BP}(\eps)$ in \eqref{eq:b_bp_bec} based on average 
``bit error probability''. In general, it gives an upper bound for the maximum recoverable burst length that the ``block error probability'' will converge to zero.
However, the simulation results in the next section suggest the tightness of upper-bound 
when $d_v\geq 3$.

\textbf{Remark 2:} In simulations of both BEC and BiAWGN channel,
we numerically observe that for any $b>b_{\BP}(\eps)$, 
$P_{\rm e}(b,\lceil s\rceil)\geq P_{\rm e}(b,s)$. It suggests that the worst case scenario is $s=\lceil s\rceil$, i.e. $z_0$ is fully erased.

\subsection*{Conditions on $w$ for $1\leq b_{\BP}(0)\leq k$:} 
Assume $b=1$ and $s=\lceil s\rceil\in \mathbb{Z}$ without further random noise ($\eps=0$), then the DE equation is simplified to
\begin{equation*}
x_{s}^{(t+1)}= (1-(1-\frac{1}{w}x_{s}^{(t)})^{d_c-1})^{d_v-1},
\end{equation*}
which is the DE equation of $(d_v,d_c)$ LDPC ensemble over a BEC with erasure probability $\frac{1}{w}$.
It implies that $P_{\rm e}(T,b=1)\neq 0$, if $\frac{1}{w}>\eps_{\BP}^{(d_v,d_c)}$, the BP threshold of the underlying $(d_v,d_c)$ LDPC ensemble. Thus, the necessary condition for
$1 \leq b_{\BP}(\eps)$ is $w\geq \lceil 1/\eps_{\BP}^{(d_v,d_c)} \rceil$.

On the other hand, $w\geq \lceil (k+1)/\eps_{\BP}^{(d_v,d_c)} \rceil$ is a sufficient condition (but not tight) for  $b_{\BP}(0)\leq k$, where $k$ is integer. 
The steps of proof are: $(i)$ 
a burst $(b_{\BP}(0),s)$ is a better channel than
a burst $(b=k+1,\lceil s\rceil)$, $(ii)$ Using DE equation for the latter burst, $\sum_{z=\lceil s\rceil}^{\lceil s\rceil+k+1}x^{(t)}_z$ can be upper-bounded by the DE equation of $(d_v,d_c)$ LDPC ensemble over BEC$(\frac{k+1}{w})$.


\section{The Burst Length in the Finite Block Length Regime}

In the limit of $M$, we observe that the decoding failure probability has a sharp transition  from zero to one as the length of burst erasures increases. In particular, $b_{\BP}(\eps)$ is the BP threshold of the combined channel of BEC$(\eps)$ and burst erasures. For a finite $M$, the decoding failure probability comprises two parts: the waterfall region for $b$ values close to $b_{\BP}(\eps)$, and the error-floor region for $b\ll b_{\BP}(\eps)$. This behaviour is illustrated in Fig.~\ref{fig:finite363}, which shows simulation results for a $\mathcal{C_R}(d_v=3,d_c=6,w=3,L=30,M)$ ensemble for $\varepsilon=0$. 
In this section, we estimate the error floor part by enumerating the size-2 \emph{stopping sets} as a function of $M$.

A subset $\mathcal{A}$ of  VNs in a code is a \emph{stopping set} if all the neighboring CNs of (the VNs in) $\mathcal{A}$ connect to $\mathcal{A}$ at least twice~\cite{Richardson-MCT08}. 
In such a case, if all VNs in $\mathcal{A}$ have been erased by the channel, then the BP decoder will fail as all the neighboring CNs are connected to at least two erased VNs. 

For simplicity, we assume here that we have only burst erasures, i.e. $\eps=0$ in Fig.~\ref{fig:b_bp36}.
The results can be later extended for the combined channel of BEC and burst erasures. We first focus on size-2 stopping sets as these dominate the performance in the error floor region~\cite{RengaswamyZSC16}.

\subsection{Size-2 Stopping Sets}

The random burst can span over multiple spatial positions because of its random starting position $sM$, and its potentially
large length $bM$. A size-2 stopping set can be formed within a single spatial
position or across coupled spatial positions. We first compute the probability of such a stopping set:


\begin{theorem}\label{thm:prob_rbc}
Consider the $C_{\mathcal{R}}(d_v,d_c,w,L,M)$ ensemble. Let $v_i$ denote a randomly chosen VN in spatial position 
$z\leq L$, and  $v_j$ denote a random VN in spatial position $z+k$, for a non-negative integer $k$ with $z+k\leq L$.
The probability that these two random VNs form a stopping set of size 2 is independent of $z$ and
amounts to
\begin{align}\label{eq:prob_rbc}
q_k=P_{\mathcal{R}}\left(1-\frac{k}{w}\right)^{d_v}, \quad k\in\{0,1,\ldots,w-1\},
\end{align}
where $P_{\mathcal{R}}$ is 
\begin{align}
P_{\mathcal{R}} \doteq  
\frac{\left(1-\frac1{d_c}\right)^{d_v}}{\sum\limits_{\ell=0}^{d_v} \binom{d_v}{\ell}
\binom{wM\frac{d_v}{d_c}-d_v}{d_v-\ell} \left(1-\frac1{d_c}\right)^\ell }. 
\label{eq:lem1_pdef}
\end{align}
For $k\geq w$, we have $q_k=0$.
\end{theorem}
\begin{proof}
Let $\mathcal{N}(v_i)$ denote the set of $d_v$ check nodes connected to VN $v_i$.
Recall that this ensemble contains no parallel edges.
From Section~\ref{sec:randomSCLDPC}, we know that $\mathcal{N}(v_i)$ can have contributions from SPs $\{z,z+1,\dots,z+w-1\}$ and 
$\mathcal{N}(v_j)$ can have contributions from SPs $\{z+k,z+k+1,\dots,z+k+w-1\}$.
A size-2 stopping set 
is formed if and only if $\mathcal{N}(v_i)=\mathcal{N}(v_j)$.
For $k\geq w$, this condition cannot be fulfilled and thus, $q_k=0$.
For $k<w$, all check nodes of $\mathcal{N}(v_i)$ must be lying in a subset $\{z+k,\dots,z+w-1\}$. As the edges of the variable nodes uniformly connect to $w$ neighboring SPs, the probability of such a selection for $v_i$ is $(\frac{w-k}{w})^{d_v}$. 
Now, we compute the the probability that $v_j$ connects exactly to the same CNs as $v_i$, 
i.e., $\mathcal{N}(v_i)=\mathcal{N}(v_j)$.
We label all the sockets of CNs in SPs $\{z+k,\dots,z+k+w-1\}$. Let $T$ denote the total number of sub-graphs from $\{v_i,v_j\}$ and let $T_{ss}$ denote the number of sub-graphs in which these VNs form a size-$2$ stopping set. Each of the CNs in $\mathcal{N}(v_i)$ has $d_c-1$ free distinct sockets.
Thus, the number of sub-graphs fulfilling $\mathcal{N}(v_i)=\mathcal{N}(v_j)$ is, 
\begin{equation*}\textstyle
T_{ss}  = d_v! (d_c-1)^{d_v},
\end{equation*}
where $d_v!$ is due to the permutation of edges and $(d_c-1)^{d_v}$ is due to the different ways of 
connecting to free sockets of $\mathcal{N}(v_i)$.
In general, $v_i$ and $v_j$ may connect to some $\ell$ common CNs, $0 \leq \ell\leq d_v$. On one hand, there are $\binom{d_v}{\ell}(d_c-1)^\ell$ socket selections for the $\ell$ common CNs. One the other hand, there are $\binom{wM\frac{d_v}{d_c}-d_v}{d_v-\ell}d_c^{d_v-\ell}$ for all other distinct $wM\frac{d_v}{d_c}-d_v$ CNs. Including $d_v!$ permutation of edges, we have,
\begin{equation*}\textstyle
T=d_v!\sum\limits_{\ell=0}^{d_v} \binom{d_v}{\ell}\binom{wM\frac{d_v}{d_c}-d_v}{d_v-\ell} (d_c-1)^\ell (d_c)^{d_v-\ell}.
\end{equation*}
We get $P_{\mathcal{R}} = \frac{T_{ss}}{T}$, simplified further to \eqref{eq:lem1_pdef}, and
hence, $q_k=\left(\frac{w-k}{w}\right)^{d_v}P_{\mathcal{R}}$.
\end{proof}

Let $\mathbb{N}_2$ denote the number of size-$2$ stopping sets in a random code instance of $C_{\mathcal{R}}(d_v,d_c,w,L,M)$ ensemble. We introduce the stopping set indicator function $U_{ij}$ with
\begin{equation*}
U_{ij} = \left\{\begin{array}{ll}
1 & \text{if VNs }v_i\text{ and }v_j\text{ form a stopping set}\\
0 & \text{otherwise}
\end{array}\right.
\end{equation*}
Then, $\mathbb{N}_2=\sum_{i=1}^{ML-1}\sum_{j=i+1}^{ML} U_{ij}$. Note that $U_{ij}$ are correlated random
variables. However, we can simply calculate the average number of size-2 stopping sets over the ensemble,
\begin{align*}
\textstyle
\mathbb{E}[\mathbb{N}_2] &=\sum_{i=1}^{ML-1}\sum_{j=i+1}^{ML} \mathbb{E}[U_{ij}]= \sum_{i=1}^{ML-1}\sum_{j=i+1}^{ML} P\{U_{ij}=1\}\\
&=\sum_{z=1}^{L}\sum_{k=0}^{w-1}\sum_{i,j=1}^M q_k \mathbbm{1}[i<kM+j,z+k\leq L]=\sum_{k=0}^{w-1} \lambda_k,
\end{align*}
where $\lambda_k$ denote the average number of size-2 stopping sets between VNs lying in two SPs with difference $k$ and
\begin{equation}
\label{eq:lambdas}
\lambda_0 = L\binom{M}{2} q_0 \hspace{2.5mm} ; \hspace{2.5mm} \lambda_k = (L-k) M^2 q_k.
\end{equation}
We see that $\lambda_k \sim O(L M^{2-d_v})$.
To verify these expectations, let us consider the $\mathcal{C}_{\mathcal{R}}(3,6,3,100,M=64)$ SC-LDPC ensemble.
By averaging over $1000$ random code instances of the ensemble,
the average number of size-$2$ stopping sets is obtained $(\lambda_0,\lambda_1,\lambda_2)\approx(0.876,0.488,0.060)$ which
is close to $(0.829,0.494,0.061)$
from \eqref{eq:lambdas}, though $M$ is rather small.

\subsection{Error Floor Estimation}

We now estimate the decoding failure (block erasure probability) when there is a random burst of length $b\ll b_\BP$ and starting bit $S$. Let $\mathbb{N}_2(S,bM)$ denote the set of size-2 stopping sets formed by VNs, $v_i$, in the burst, i.e., $i\in[S,S+bM]$. BP decoding fails if these VNs are erased. Thus,
\begin{equation}
P_\B(b)\geq \mathbb{P}\{\mathbb{N}_2(S,bM)\geq 1\}\overset{(i)}{\approx} \mathbb{E}[\mathbb{N}_2(S,bM)].
\label{eq:lowerbound}
\end{equation} 
There are two approaches to justify $(i)$. The first approach is to lower-bound $\mathbb{P}\{\mathbb{N}_2(S,bM)\geq 1\}$ using  the second moment method and to show 
that the bound has a vanishing gap (in $M$) to 
$\mathbb{E}[\mathbb{N}_2(S,bM)]$. 
We applied this method in~\cite{RengaswamyZSC16} for the particular choice of $b=1$ and $S=kM+1$, $k\in[0,L-1]$. 
An alternative is to use standard arguments~\cite[App. C]{Richardson-MCT08} to approximate the distribution of size-2 stopping sets by a
joint Poisson distribution. The decoding error then corresponds approximately to the average number of stopping sets.

The starting bit $S$ is chosen uniformly among bits $[1,LM-bM+1]$. We can write $S=(z_0-1)M+j$, for $z_0\in\mathbb{N}$ and some integer $1\leq j\leq M$. Then,
\begin{align}
\textstyle
\mathbb{E}&[\mathbb{N}_2(S,bM)] = \frac{1}{LM-bM+1}\sum_{k=1}^{\mathclap{LM-bM+1}} 
\mathbb{E}[\mathbb{N}_2(S,bM)\mid S=k]\nonumber\\
&\overset{(i)}{\gtrapprox} \frac{M}{(L-b)M+1}\sum_{z_0=1}^{L-\lceil b\rceil}\frac{1}{M}\sum_{j=1}^{M} \mathbb{E}[\mathbb{N}_2((z_0-1)M+j,bM)]\nonumber \\
%
&\overset{(ii)}{=}\frac{L-\lceil b\rceil}{(L-b)M+1} 
\sum_{j=1}^{M} \mathbb{E}[\mathbb{N}_2(S=j,bM)]\nonumber\\
&\approx  \frac{1}{M}\sum_{j=1}^{M} \mathbb{E}[\mathbb{N}_2(S=j,bM)]\nonumber\\
&\overset{(iii)}{=}\frac{1}{M}\sum_{S=1}^{M} \sum_{z=1}^{\lceil b\rceil+1}\left(\binom{m_z}{2}q_0+\sum_{k=1}^{w-1} m_z m_{z+k}q_k\right).
\textstyle
\label{eq:errorfloor}
\end{align}
where $(i)$ is because we neglect a small contribution $(O(\frac{1}{L}))$ of $S>(L-\lceil b\rceil)M$ for non-integer $b$. We have $(ii)$ as $\frac{1}{M}\sum_{j=1}^{M} \mathbb{E}[\mathbb{N}_2((z_0-1)M+j,bM)]$ is identical for different $z_0$. Let us justify $(iii)$: for a given starting bit $1\leq S\leq M$, the number of erased VNs in SP $z$ is $m_z$ defined in Section~\ref{sec:asymp}. We have $(iii)$ by summing the average number of size-2 stopping sets formed between erased VNs in all pairs of SP $z$ and $z+k$.

We plot the decoding failure probability of $(3,6,w,L,M)-$SC-LDPC codes for different finite values of $M$ and for $w=3,4$ in Fig.~\ref{fig:finite363} $(a)$-$(b)$. 
For each pair of $M$ and $b$, we choose a random instance from the code ensemble and generate a random 
burst with length $bM$. The decoding failure probability, $P_\B$, is averaged over 
all trials until 400 decoding failures occur. We repeat the same experiment for $(4,8,w=4,L,M)-$SC-LDPC codes, depicted in Fig.~\ref{fig:finite363}-$(c)$.
We also plot the error floor estimation \eqref{eq:errorfloor} for each $M$.

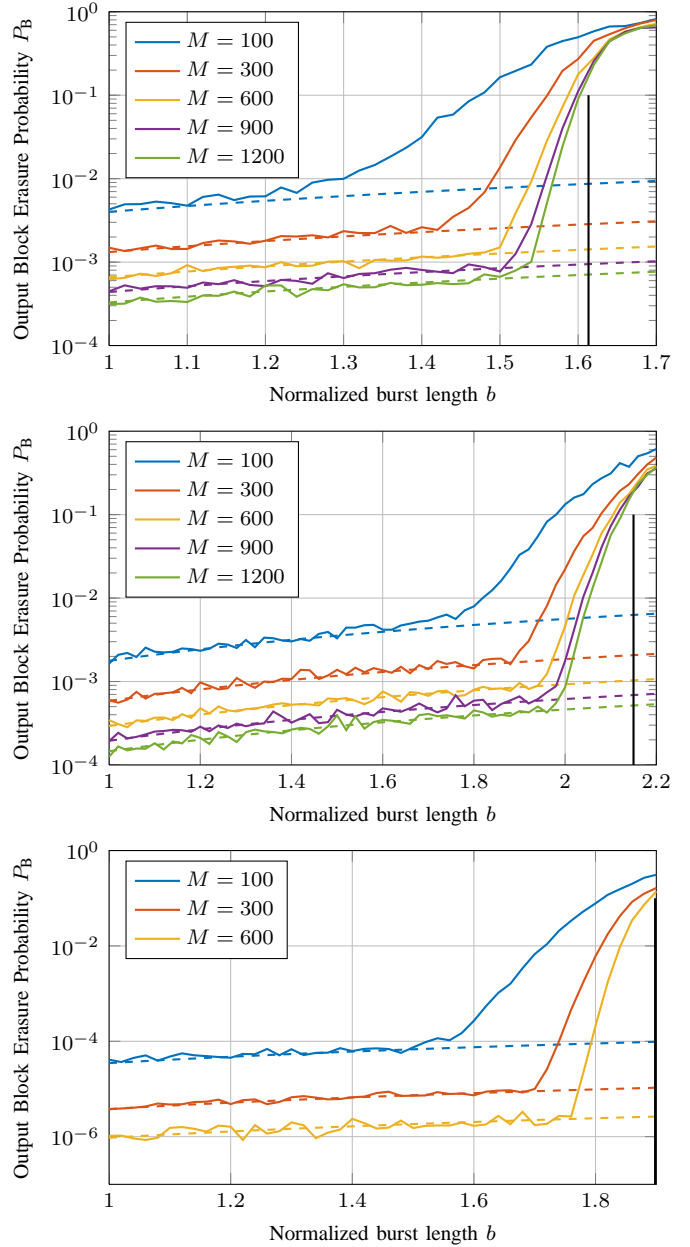
\begin{figure}[tb!]
\begin{tikzpicture}[every node/.style={font=\footnotesize}]
\begin{scope}
\begin{axis}[%
width=\columnwidth,
height=0.68\columnwidth,
every axis/.append style={font=\small},
every x tick label/.append style={font=\footnotesize},
every y tick label/.append style={font=\footnotesize},
xlabel={Normalized burst length $b$},
xmajorgrids,
ymode=log,
xmin=1,
xmax=1.7,
ymin=1e-4,
ymax=1,
yminorticks=true,
ylabel={Output Block Erasure Probability $P_{\text{B}}$},
ymajorgrids,
legend pos=north west,
legend style={legend cell align=left,align=left,draw=black},
]

\definecolor{mycolor1}{rgb}{0.00000,0.44700,0.74100}
\definecolor{mycolor2}{rgb}{0.85000,0.32500,0.09800}
\definecolor{mycolor3}{rgb}{0.92900,0.69400,0.12500}
\definecolor{mycolor4}{rgb}{0.49400,0.18400,0.55600}
\definecolor{mycolor5}{rgb}{0.46600,0.67400,0.18800}

\addplot [color=mycolor1,thick,solid,mark=none]
table[x=bval,y=SimM100] {result_RBC_PB_vs_b_dv3_dc6_w3.txt};
\addlegendentry{$M=100$};
\addplot [color=mycolor1,thick,dashed,mark=none,forget plot]
table[x=bval,y=BoundM100] {result_RBC_PB_vs_b_dv3_dc6_w3.txt};

\addplot [color=mycolor2,thick,solid,mark=none]
table[x=bval,y=SimM300] {result_RBC_PB_vs_b_dv3_dc6_w3.txt};
\addlegendentry{$M=300$};
\addplot [color=mycolor2,thick,dashed,mark=none,forget plot]
table[x=bval,y=BoundM300] {result_RBC_PB_vs_b_dv3_dc6_w3.txt};

\addplot [color=mycolor3,thick,solid,mark=none]
table[x=bval,y=SimM600] {result_RBC_PB_vs_b_dv3_dc6_w3.txt};
\addlegendentry{$M=600$};
\addplot [color=mycolor3,thick,dashed,mark=none,forget plot]
table[x=bval,y=BoundM600] {result_RBC_PB_vs_b_dv3_dc6_w3.txt};

\addplot [color=mycolor4,thick,solid,mark=none]
table[x=bval,y=SimM900] {result_RBC_PB_vs_b_dv3_dc6_w3.txt};
\addlegendentry{$M=900$};
\addplot [color=mycolor4,thick,dashed,mark=none,forget plot]
table[x=bval,y=BoundM900] {result_RBC_PB_vs_b_dv3_dc6_w3.txt};

\addplot [color=mycolor5,thick,solid,mark=none]
table[x=bval,y=SimM1200] {result_RBC_PB_vs_b_dv3_dc6_w3.txt};
\addlegendentry{$M=1200$};
\addplot [color=mycolor5,thick,dashed,mark=none,forget plot]
table[x=bval,y=BoundM1200] {result_RBC_PB_vs_b_dv3_dc6_w3.txt};

\addplot [color=black,thick,solid,mark=none,forget plot] coordinates {(1.6133,1e-4) (1.6133,0.1)} node [above,sloped,pos=0.5] {$b_{\text{BP}}$};

\end{axis}
\end{scope}

\begin{scope}[yshift=-.63\columnwidth]

\begin{axis}[%
width=\columnwidth,
height=0.68\columnwidth,
every axis/.append style={font=\small},
every x tick label/.append style={font=\footnotesize},
every y tick label/.append style={font=\footnotesize},
xlabel={Normalized burst length $b$},
xmajorgrids,
ymode=log,
xmin=1,
xmax=2.2,
ymin=1e-4,
ymax=1,
yminorticks=true,
ylabel={Output Block Erasure Probability $P_{\text{B}}$},
ymajorgrids,
legend pos=north west,
legend style={legend cell align=left,align=left,draw=black},
]

\definecolor{mycolor1}{rgb}{0.00000,0.44700,0.74100}
\definecolor{mycolor2}{rgb}{0.85000,0.32500,0.09800}
\definecolor{mycolor3}{rgb}{0.92900,0.69400,0.12500}
\definecolor{mycolor4}{rgb}{0.49400,0.18400,0.55600}
\definecolor{mycolor5}{rgb}{0.46600,0.67400,0.18800}

\addplot [color=mycolor1,thick,solid,mark=none]
table[x=bval,y=SimM100] {result_RBC_PB_vs_b_dv3_dc6_w4.txt};
\addlegendentry{$M=100$};
\addplot [color=mycolor1,thick,dashed,mark=none,forget plot]
table[x=bval,y=BoundM100] {result_RBC_PB_vs_b_dv3_dc6_w4.txt};

\addplot [color=mycolor2,thick,solid,mark=none]
table[x=bval,y=SimM300] {result_RBC_PB_vs_b_dv3_dc6_w4.txt};
\addlegendentry{$M=300$};
\addplot [color=mycolor2,thick,dashed,mark=none,forget plot]
table[x=bval,y=BoundM300] {result_RBC_PB_vs_b_dv3_dc6_w4.txt};

\addplot [color=mycolor3,thick,solid,mark=none]
table[x=bval,y=SimM600] {result_RBC_PB_vs_b_dv3_dc6_w4.txt};
\addlegendentry{$M=600$};
\addplot [color=mycolor3,thick,dashed,mark=none,forget plot]
table[x=bval,y=BoundM600] {result_RBC_PB_vs_b_dv3_dc6_w4.txt};

\addplot [color=mycolor4,thick,solid,mark=none]
table[x=bval,y=SimM900] {result_RBC_PB_vs_b_dv3_dc6_w4.txt};
\addlegendentry{$M=900$};
\addplot [color=mycolor4,thick,dashed,mark=none,forget plot]
table[x=bval,y=BoundM900] {result_RBC_PB_vs_b_dv3_dc6_w4.txt};

\addplot [color=mycolor5,thick,solid,mark=none]
table[x=bval,y=SimM1200] {result_RBC_PB_vs_b_dv3_dc6_w4.txt};
\addlegendentry{$M=1200$};
\addplot [color=mycolor5,thick,dashed,mark=none,forget plot]
table[x=bval,y=BoundM1200] {result_RBC_PB_vs_b_dv3_dc6_w4.txt};

\addplot [color=black,thick,solid,mark=none,forget plot] coordinates {(2.15,1e-4) (2.15,0.1)} node [above,sloped,pos=0.1] {$b_{\text{BP}}$};

\end{axis}
\end{scope}

\begin{scope}[yshift=-1.26\columnwidth]
\begin{axis}[%
width=\columnwidth,
height=0.68\columnwidth,
every axis/.append style={font=\small},
every x tick label/.append style={font=\footnotesize},
every y tick label/.append style={font=\footnotesize},
xlabel={Normalized burst length $b$},
xmajorgrids,
ymode=log,
xmin=1,
xmax=1.9,
ymin=1e-7,
ymax=1,
yminorticks=true,
ylabel={Output Block Erasure Probability $P_{\text{B}}$},
ymajorgrids,
legend pos=north west,
legend style={legend cell align=left,align=left,draw=black},
]

\definecolor{mycolor1}{rgb}{0.00000,0.44700,0.74100}
\definecolor{mycolor2}{rgb}{0.85000,0.32500,0.09800}
\definecolor{mycolor3}{rgb}{0.92900,0.69400,0.12500}
\definecolor{mycolor4}{rgb}{0.49400,0.18400,0.55600}
\definecolor{mycolor5}{rgb}{0.46600,0.67400,0.18800}

\addplot [color=mycolor1,thick,solid,mark=none]
table[x=bval,y=SimM100] {result_RBC_PB_vs_b_dv4_dc8_w4.txt};
\addlegendentry{$M=100$};
\addplot [color=mycolor1,thick,dashed,mark=none,forget plot]
table[x=bval,y=BoundM100] {result_RBC_PB_vs_b_dv4_dc8_w4.txt};

\addplot [color=mycolor2,thick,solid,mark=none]
table[x=bval,y=SimM300] {result_RBC_PB_vs_b_dv4_dc8_w4.txt};
\addlegendentry{$M=300$};
\addplot [color=mycolor2,thick,dashed,mark=none,forget plot]
table[x=bval,y=BoundM300] {result_RBC_PB_vs_b_dv4_dc8_w4.txt};

\addplot [color=mycolor3,thick,solid,mark=none]
table[x=bval,y=SimM600] {result_RBC_PB_vs_b_dv4_dc8_w4.txt};
\addlegendentry{$M=600$};
\addplot [color=mycolor3,thick,dashed,mark=none,forget plot]
table[x=bval,y=BoundM600] {result_RBC_PB_vs_b_dv4_dc8_w4.txt};

\addplot [color=black,thick,solid,mark=none,forget plot] coordinates {(1.898,1e-7) (1.898,0.1)} node [above,sloped,pos=0.1] {$b_{\text{BP}}$};

\end{axis}

\end{scope}

\end{tikzpicture}%
\caption{\label{fig:finite363} Simulation results for $(a)$ $d_v=3$, $d_c=6$ with $w=3$, for $(b)$ $d_v=3$, $d_c=6$ with $w=4$, and for $(c)$ $d_v=4$, $d_c=8$ with $w=4$. Solid lines represent simulation results and dashed lines the error bound.}
\end{figure}

These figures show that for $b<b_{\BP}(\eps=0)$, the error floor is well estimated by \eqref{eq:errorfloor} even for small $M=100$. It implies that the size-2 stopping sets are the main cause of decoding error. 
We also observe that the decoding error increases very fast for $b$ close to $b_{\BP}(\eps=0)$, given in Fig.~\ref{fig:b_bp36}.
For larger $M$, the waterfall region is sharper around the threshold $b_{\BP}(\eps=0)$.


%
%

\section{Conclusion}
\label{sec:conclude}

In this paper, we have investigated the performance of spatially coupled LDPC codes when the transmission is affected by a single burst of erasures per codeword. Such a burst erasure can model different scenarios, e.g., the outage of a node in distributed transmission. We have derived an expression for density evolution and shown numerically that the maximum correctable burst length depends on the channel that affects the bits not erased by the burst and the code parameters. Depending on the expected burst, different parameters may be selected to design a code. The correctable burst length is practically independent of the transmission channel of the other bits. Furthermore, we have given expressions for the error floor that remains after correction. We have successfully verified all results in a simulation example.

\bibliographystyle{IEEEtran} 
\bibliography{IEEEabrv,refer}

\begin{thebibliography}{10}
\providecommand{\url}[1]{#1}
\csname url@samestyle\endcsname
\providecommand{\newblock}{\relax}
\providecommand{\bibinfo}[2]{#2}
\providecommand{\BIBentrySTDinterwordspacing}{\spaceskip=0pt\relax}
\providecommand{\BIBentryALTinterwordstretchfactor}{4}
\providecommand{\BIBentryALTinterwordspacing}{\spaceskip=\fontdimen2\font plus
\BIBentryALTinterwordstretchfactor\fontdimen3\font minus
  \fontdimen4\font\relax}
\providecommand{\BIBforeignlanguage}[2]{{%
\expandafter\ifx\csname l@#1\endcsname\relax
\typeout{** WARNING: IEEEtran.bst: No hyphenation pattern has been}%
\typeout{** loaded for the language `#1'. Using the pattern for}%
\typeout{** the default language instead.}%
\else
\language=\csname l@#1\endcsname
\fi
#2}}
\providecommand{\BIBdecl}{\relax}
\BIBdecl

\bibitem{Lentmaier-ita09}
M.~Lentmaier, G.~P. Fettweis, K.~Zigangirov, and D.~J. {Costello, Jr.},
  ``Approaching capacity with asymptotically regular {LDPC} codes,'' in
  \emph{Proc. ITA}, 2009.

\bibitem{Kudekar-it11}
S.~Kudekar, T.~Richardson, and R.~Urbanke, ``Threshold saturation via spatial
  coupling: Why convolutional {LDPC} ensembles perform so well over the
  {BEC},'' \emph{{IEEE} Trans. Inf. Theory}, vol.~57, no.~2, 2011.

\bibitem{Kudekar-it13}
S.~Kudekar, T.~Richardson, and R.~L. Urbanke, ``Spatially coupled ensembles
  universally achieve capacity under belief propagation,'' \emph{{IEEE} Trans.
  Inf. Theory}, vol.~59, no.~12, pp. 7761--7813, 2013.

\bibitem{Jule-isit13}
A.~Jule and I.~Andriyanova, ``Performance bounds for spatially-coupled {LDPC}
  codes over the block erasure channel,'' in \emph{Proc. IEEE ISIT}, July 2013,
  pp. 1879--1883.

\bibitem{Iyengar-icc10}
A.~Iyengar, M.~Papaleo, G.~Liva, P.~Siegel, J.~Wolf, and G.~Corazza,
  ``Protograph-based {LDPC} convolutional codes for correlated erasure
  channels,'' in \emph{Proc. IEEE ICC}, May 2010, pp. 1--6.

\bibitem{Mori-corr15}
\BIBentryALTinterwordspacing
H.~Mori and T.~Wadayama, ``Band splitting permutations for spatially coupled
  {LDPC} codes enhancing burst erasure immunity,'' \emph{arXiV}, 2015.
\BIBentrySTDinterwordspacing

\bibitem{ulHassan-isit14}
 I.~Andriyanova, N.~{ul Hassan}, M.~Lentmaier, and G.~P. Fettweis, ``SC-LDPC codes over the block-fading channel: Robustness to a synchronisation offset,'' in \emph{Proc.
  2015 IEEE BlackSeaCom}, May 2015, pp. 97--101.

\bibitem{ulHassan-itw15}
N.~{ul Hassan}, I.~Andriyanova, M.~Lentmaier, and G.~P. Fettweis, ``Protograph
  design for spatially-coupled codes to attain an arbitrary diversity order,''
  in \emph{Proc. ITW}, Jeju City, South Korea, Oct. 2015.

\bibitem{Jardel-comnet15}
F.~Jardel, J.~J. Boutros, M.~Sarkiss, and G.~{Rekaya-Ben Othman}, ``Spatial
  coupling for distributed storage and diversity applications,'' in \emph{Proc.
  IEEE ComNet},
  Hammamet, Tunisia, Nov. 2015.

\bibitem{Boutros-tit10}
J.~J. Boutros, A.~G.~I. Fabregas, E.~Biglieri, and G.~Z{\'e}mor, ``Low-density
  parity-check codes for nonergodic block-fading channels,'' \emph{{IEEE} Trans. Inf. Theory}, vol.~56, no.~9,
  pp. 4286--4300, 2010.

\bibitem{RengaswamyZSC16}
N.~Rengaswamy, L.~Schmalen, and V.~Aref, ``{On the Burst Erasure Correctability
  of Spatially Coupled {LDPC} Ensembles},'' in \emph{International Zurich
  Seminar on Communications}, Zurich, CH, Mar. 2016.

\bibitem{Richardson-MCT08}
T.~Richardson and R.~Urbanke, \emph{Modern {C}oding {T}heory}.\hskip 1em plus
  0.5em minus 0.4em\relax Cambridge University Press, 2008.

\end{thebibliography}

\end{document}